\documentclass[10pt,journal, onecolumn]{IEEEtran}
\usepackage{filecontents}

\usepackage{pgf,tikz}
\usetikzlibrary{arrows}

\usepackage{amssymb}
\usepackage{amsmath}
\usepackage{amsthm}
\usepackage{multirow}
\usepackage{units}
\usepackage{mathtools}
\usepackage{hyperref}
\usepackage{cite}

\usepackage{hhline}

\DeclarePairedDelimiter{\floor}{\lfloor}{\rfloor}
\DeclarePairedDelimiter{\bigfloor}{\Bigl\lfloor}{\Bigr\rfloor}

\newtheorem{theorem}{Theorem}
\newtheorem{definition}{Definition}

\newtheorem{lemma}{Lemma}
\newtheorem{corollary}{Corollary}

\newtheorem{construction}{Construction}
\newtheorem{example}{Example}

\numberwithin{subcase}{case}

\newtheorem{remark}{Remark}

\newcommand{\cl}[1]{\mathcal{#1}}

\newcommand{\Fq}{\mathbb{F}_{q}}

\newcommand{\bC}{\mathbb{C}}

\newcommand{\bF}{\mathbb{F}}

\newcommand{\bN}{\mathbb{N}}

\newcommand{\cC}{\mathcal{C}}

\newcommand{\cP}{\mathcal{P}}

\newcommand{\cZ}{\mathcal{Z}}

\newcommand{\qbin}[3]{{#1 \brack #2}_{#3}}
\newcommand{\grsmn}[3]{\cl{G}_{#1}\left(#2,#3\right)}
\newcommand{\Span}[1]{{\left\langle {#1} \right\rangle}}

\DeclareMathOperator{\rank}{rank}

\DeclareMathOperator{\Gab}{Gab}
\DeclareMathOperator{\RS}{RS}
\newcommand{\I}{{I}}
\newcommand{\C}{{C}}

\newcommand{\X}{{X}}
\newcommand{\Y}{{Y}}
\newcommand{\liftmap}[1]{\mathcal{I}\left(#1\right)}

\begin{document}

\title{Some Gabidulin Codes Cannot be\\List Decoded Efficiently at any Radius}

\author{Netanel Raviv, \textit{Student Member, IEEE}, and \IEEEauthorblockN{Antonia Wachter-Zeh}\textit{, Member, IEEE}
\thanks{
	The work of Netanel Raviv was supported in part by the Israeli Science Foundation (ISF), Jerusalem, Israel, under Grant no.~10/12, and by the IBM Ph.D. fellowship. The work of Antonia Wachter-Zeh was supported by a Minvera Postdoctoral Fellowship from the Max-Planck Society and by the European Union's Horizon 2020 research
	and innovation programme under the Marie Sklodowska-Curie grant agreement No 655109.
	
	The authors are with the department of Computer Science, Technion~-- Israel Institute of Technology, Haifa 3200003, Israel (e-mail:
	\textit{\{netanel, antonia\}@cs.technion.ac.il).}
	
	Parts of this work have been presented at the \emph{IEEE International Symposium on Information Theory (ISIT) 2015, HongKong, China}~\cite{RavivWachterzeh-ISITListDecoding}.}
}

\maketitle

\begin{abstract}
Gabidulin codes can be seen as the rank-metric equivalent of Reed--Solomon codes. It was recently proven, using subspace polynomials,  that Gabidulin codes cannot be list decoded beyond the so-called Johnson radius. In another result, cyclic subspace codes were constructed by inspecting the connection between subspaces and their subspace polynomials. In this paper, these subspace codes are used to prove two bounds on the list size in decoding certain Gabidulin codes. The first bound is an existential one, showing that exponentially-sized lists exist for codes with specific parameters. The second bound presents exponentially-sized lists explicitly, for a different set of parameters. Both bounds rule out the possibility of efficiently list decoding several families of Gabidulin codes for any radius beyond half the minimum distance. Such a result was known so far only for non-linear rank-metric codes, and not for Gabidulin codes. Using a standard operation called lifting, identical results also follow for an important class of constant dimension subspace codes.
\end{abstract}

\begin{IEEEkeywords}
Rank-metric codes, Gabidulin codes, list decoding, subspace polynomials, subspace codes.
\end{IEEEkeywords}

\IEEEpeerreviewmaketitle

\section{Introduction}\label{section:introduction}
Rank-metric codes have recently attracted increasing interest due to their application to error correction in random network coding~\cite{ARankMetricApproach} where they can be used to construct constant dimension subspace codes.
Further applications of codes in the rank metric include cryptography \cite{Gabidulin1991Ideals,Loidreau2010Designing}, space-time coding \cite{Lusina2003Maximum,Lu2004Generalized} and distributed storage systems~\cite{SilbersteinRawat-OptimallyLocallyRepairable_2013,SilbersteinRawatVish-ErrorResilDistributedStorage_2012}.

For a prime power $q$, let $\bF_q$ be the field with $q$ elements. For an integer $n$, let $\bF_{q^n}$ be the extension field of degree $n$ of $\bF_q$ (which may be seen as the vector space of dimension $n$ over $\bF_q$, denoted by $\bF_q^n$), and $\bF_{q^n}^*\triangleq \bF_{q^n}\setminus\{0\}$. For $m\ge n$, a rank-metric code is a set of $m\times n$ matrices over $\bF_q$, or alternatively, a set of vectors of length $n$ over the extension field $\bF_{q^m}$, where the distance between two matrices is the rank of their difference. {The \textit{rate} of a rank metric code of size~$M$ is $\frac{\log _{q} M}{mn}$}.  \textit{Gabidulin} codes, introduced by \cite{Delsarte,TheoryOfCodesWithMaximumRankDistance,MaximumRankArrayCodes}, may be seen as the rank-metric equivalent of Reed--Solomon codes. These codes are defined as evaluations of \textit{linearized polynomials} (see below) of bounded degree at a given set of linearly independent evaluation points. We note that Gabidulin codes, and rank-metric codes in general, may be defined for any $m\ge n$, while our results only apply for the case $n$ divides $m$ (and in some cases, when $n+1$ divides $m$ by puncturing). In particular, our results apply for~$n=m$.

Given a word $w\in\bF_{q^m}^n$ (or alternatively, a matrix ${w\in\bF_q^{m\times n}}$), a \textit{list decoding} algorithm outputs all Gabidulin codewords that are inside a ball of radius $\tau$, centered at $w$, where $\tau$ is possibly larger than the unique decoding radius of the code. For a given code, a natural question to ask is: for which values of $\tau$ can list decoding be done efficiently? List decoding of rank-metric codes and Gabidulin codes was recently studied in \cite{Ding,GuruswamiExplicit,BoundsOnListDecodingOfRankMetric}. In \cite{BoundsOnListDecodingOfRankMetric}, it was shown that Gabidulin codes cannot be list decoded beyond the Johnson radius. 
This result was generalized to any rank-metric code by~\cite{Ding}. {When $m$ is sufficiently large,~\cite{Ding} also showed that with high probability a random rank-metric code can be efficiently list decoded}. 
Further, it was shown in~\cite{BoundsOnListDecodingOfRankMetric} that there is no Johnson-like polynomial upper bound on the list size since there exists a non-linear rank-metric code with exponentially growing list size for any radius greater than the unique decoding radius. 
In~\cite{GuruswamiExplicit}, an explicit subcode of a Gabidulin code was shown to be efficiently list decodable. In addition,  \cite{Ding,GuruswamiExplicit}, and \cite{BoundsOnListDecodingOfRankMetric} have noted that it is not known if Gabidulin codes themselves can be efficiently list decoded beyond the unique decoding radius. In this paper, it is shown that the answer to this question is negative.

Clearly, if there exists a word $w\in\bF_{q^m}^n$ with exponentially many Gabidulin codewords in a radius $\tau$ around it, then efficient list decoding is not possible for this radius. This combinatorial technique was used in \cite{SubspacePolynomialsAndLimits} to show the limits of list decoding of Reed--Solomon codes, and in \cite{BoundsOnListDecodingOfRankMetric} to show the limits of list decoding of Gabidulin codes. 

The main tool in \cite{SubspacePolynomialsAndLimits,BoundsOnListDecodingOfRankMetric} is subspace polynomials, which are a special type of linearized polynomials. Linearized polynomials, defined by Ore \cite{OnASpecialClassOfPolynomials}, are polynomials of the form
\[
P(x)= a_r \cdot x^{[r]}+\cdots +a_1\cdot x^{[1]}+a_0\cdot x,
\]
where $[i]\triangleq q^i$ and the coefficients are in the finite field~$\bF_{q^n}$ for some given $n$. For a linearized polynomial $P$, define the $q$-degree of $P$ as $\deg_q P\triangleq r= \log_q \deg P$. Using the isomorphism between $\bF_{q^n}$ and $\bF_q^n$, every linearized polynomial may be seen as an $\bF_q$-linear function from $\bF_q^n$ to itself \cite[Chapter~4, p.~108]{FiniteFields}, that is,  for every~$\alpha,\beta\in\bF_q$ and $u,v\in\bF_{q^n}$, each linearized polynomial $P$ satisfies $P(\alpha v+ \beta u)=\alpha P(v)+\beta P(u)$. A subspace polynomial is defined as follows.
\begin{definition}\label{definition:subspacePoly}\cite{Cyclic,AffineDispersers,SubspacePolynomialsAndLimits,ConstructingHighOrderElements,BoundsOnListDecodingOfRankMetric}
A monic linearized polynomial $P$ is called a \textit{subspace polynomial} with respect to $\bF_{q^n}$ if it satisfies the following equivalent conditions:

\begin{enumerate}
\item[A1.] $P$ divides $x^{[n]}-x$.
\item[A2.] $P$ splits completely over $\bF_{q^n}$ and all its roots have multiplicity one.
\item[A3.] \label{item:unique} {For some $0\le r\le n$, there exists an $r$-dimensional subspace $V$ of $\bF_{q^n}$ such that $P(x)=\prod_{v\in V}(x-v)$.}
\end{enumerate}
\end{definition} 

By A3, each subspace $V$ corresponds to a unique subspace polynomial, denoted $P_V$. Subspace polynomials are an efficient method of representing subspaces,
from which one can directly deduce certain properties of the subspace which are not evident in some
other representations. These objects were studied in the past for various other purposes, e.g., construction of affine dispersers~\cite{AffineDispersers}, finding an element of high multiplicative order in a finite field \cite{ConstructingHighOrderElements}, and construction of cyclic subspace codes~\cite{Cyclic}. Albeit this wide range of applications, not much is known about the coefficients of subspace polynomials and their connection to the properties of the subspace. 

It is known that all roots of every linearized polynomial have the same multiplicity, which is an integer power of~$q$, and these roots form a subspace in the extension field \cite[Theorem~3.50, p.~108]{FiniteFields}. Therefore, any monic linearized polynomial is a power of a subspace polynomial with respect to its splitting field. However, the structure of the coefficients of subspace polynomials, compared to other linearized polynomials of the same degree, is generally not known. A partial answer to this question was given by \cite{Cyclic}, and we use similar techniques 
to show limits of list decoding of Gabidulin codes.

Ben-Sasson et al. \cite{SubspacePolynomialsAndLimits} proved that a given set of subspace polynomials with mutual top coefficients provides an upper bound on the list decoding radius of Reed--Solomon codes. A counting argument was later applied in order to show that such large sets of subspace polynomials do exist. A similar technique was used in \cite{BoundsOnListDecodingOfRankMetric} to show the limits of list decoding of Gabidulin codes. In the sequel, the existence of a set of subspaces whose polynomials have a larger agreement is proved (Theorem~\ref{theorem:GabidulinCounting}). This set is a subset of a subspace code by \cite{Cyclic}. Furthermore, \textit{explicit} dense sets of words in a Gabidulin code are provided (Theorem~\ref{theorem:GabidulinExplicit}). Both bounds are used to show that the respective families of Gabidulin codes cannot be list decoded efficiently \textit{at all}. That is, there exist received words that have exponentially many codewords around them, already for a radius which is only larger than the unique decoding radius by one (Examples~\ref{example:noLD} and~\ref{example:ronny}, and Theorem~\ref{theorem:GabidulinExplicit}). Due to a technical limitation of our techniques, the presented families have rate at least $\frac{1}{5}$. 

Subspace codes have attracted an increasing interest recently due to their application in error correction in random network coding~\cite{KandK}.
It is widely known that rank-metric codes are deeply connected to constant dimension subspace codes through an operation called lifting~\cite{ConstantRankCodesAndTheir,ARankMetricApproach}. This operation preserves the distance and the cardinality of the original rank-metric code. An important family of nearly optimal constant dimension subspace codes are \textit{lifted Gabidulin codes} (that are a special case of the so-called K\"{o}tter and Kschischang codes~\cite{KandK}), which result from Gabidulin codes by lifting (see Definition~\ref{def:lifting_matrixcode}). List decoding of subspace codes was extensively studied in recent years. In particular, several variants and subcodes of the K\"{o}tter and Kschischang codes were shown to be efficiently list decodable~(e.g.,~\cite{Ding,GuruswamiExplicit,GuruswamiExplicit2,LD1,LD2} and references therein), and bounds equivalent to~\cite{BoundsOnListDecodingOfRankMetric} were discussed in~\cite{RosenthalSubspace}. Our results about Gabidulin codes also apply for lifted Gabidulin codes, and thus we get families of subspace codes that cannot be list decoded efficiently at any radius. Our techniques may also be used for showing limits to list decoding of Reed--Solomon codes, but the resulting bounds are too weak to provide any useful insight.

These results reveal a significant difference in list decoding Gabidulin and Reed--Solomon codes, although the definitions of these code classes strongly resemble each other. Namely, Reed--Solomon codes can be efficiently list decoded up to the Johnson radius (with the Guruswami-Sudan algorithm~\cite{Guruswami-Sudan}), whereas we have just proven that (some classes of) Gabidulin codes cannot be list decoded efficiently at all.

The rest of the paper is organized as follows. Notations for subspace codes and the subspace code from~\cite{Cyclic} will be described in Section~\ref{section:preliminaries}, together with the required background on cyclic shifts of subspaces and $q$-associates of polynomials. In Section~\ref{section:setOfSSP}, the code from Section~\ref{section:preliminaries} is used to prove the existence of a certain set of subspace polynomials, and the notion of $q$-associates is used to show an explicit set of another type of subspace polynomials. The improved bounds on list decodability of Gabidulin codes are discussed in Section~\ref{section:improved}, implications about subspace codes are discussed in Section~\ref{section:subspaceCodes}, and conclusions are given in Section~\ref{section:conclusions}. A discussion about the inapplicability of our techniques to list decodability of Reed--Solomon codes appears in~\nameref{appendix:A}.
\section{Preliminaries}\label{section:preliminaries}
The set $\grsmn{q}{n}{r}$, called the \textit{Grassmannian}, is the set of all subspaces of dimension $r$ ($r$-subspaces, in short) of~$\bF_{q^n}$. The size of $\grsmn{q}{n}{r}$ is given by the Gaussian coefficient $\qbin{n}{r}{q}\triangleq\prod_{i=0}^{r-1}\frac{q^{n-i}-1}{q^{i+1}-1}$, which satisfies $q^{r(n-r)}\le \qbin{n}{r}{q}\le 4 q^{r(n-r)}$~\cite{ConstantRankCodesAndTheir}. A constant dimension \textit{subspace code}~\cite{KandK} is a subset of $\grsmn{q}{n}{r}$ under the \textit{subspace metric} $d_S(U,V)=\dim U+\dim V-2\dim(U\cap V)$. 

An extensively used concept in this paper is \textit{cyclic shifts} of subspaces, defined as follows.

\begin{definition}
For $V\in\grsmn{q}{n}{r}$ and $\alpha\in\bF_{q^n}^*$ let $\alpha V\triangleq\{\alpha v\vert v\in V\}$.
\end{definition}
The set $\alpha V$, which is clearly a subspace of the same dimension as $V$, is called a \textit{cyclic shift} of~$V$. Cyclic shifts were shown to be useful for constructing subspace codes~\cite{Cyclic,ErrorCorrectingCodesInProjectiveSpace}. The set of all cyclic shifts of $V\in\grsmn{q}{n}{r}$ is called the \textit{orbit} of $V$, and its size is $\frac{q^n-1}{q^t-1}$ for some integer $t$ which divides~$n$. The size of the orbit and the structure of its subspace polynomials can be derived by inspecting the subspace polynomial of $V$, as shown in the following lemmas.

\begin{lemma}\cite[Lemma~5]{Cyclic} \label{lemma:cyclicShiftPoly}
If $V\in \grsmn{q}{n}{r}$ and $\alpha \in \bF_{q^n}^{*}$ then $P_{\alpha V}(x) = \alpha ^{[r]}\cdot P_V (\alpha^{-1}x)$. That is, if $P_V(x) = x^{[r]} +\sum_{j=0}^{r-1} \alpha_j x^{[j]}$ then $P_{\alpha V}(x)=x^{[r]} +\sum_{j=0}^{r-1} \alpha^{[r]-[j]} \alpha_j x^{[j]}$.
\end{lemma}

\begin{lemma}\cite[Corollary~3]{Cyclic} \label{lemma:distictCyclic}
Let $V\in\grsmn{q}{n}{r}$ and $P_V(x)=x^{[r]}+\sum_{j=0}^{r-1} \alpha_j x^{[j]}$. If $\alpha_s \ne 0$ for some $s\in \left\{1,\ldots,r-1\right\}$ and $\gcd(s,n)=t$, then $V$ has at least $\frac{q^n-1}{q^t-1}$ distinct cyclic shifts.
\end{lemma}

In \cite{Cyclic} it is shown that subspaces in $\grsmn{q}{n}{r}$, that may be considered as subspaces over a subfield of $\bF_{q^n}$ which is larger than $\bF_q$, admit a unique subspace polynomial structure. In what follows we cite the essentials from~\cite{Cyclic}. 
{For an integer $g$ such that $g|\gcd(n,r)$, let $h$ be any $\bF_{q^g}$ isomorphism between $\bF_{q^g}^{n/g}$ and $\bF_{q^n}$, and notice that for all $u,v\in\bF_{q^g}^{n/g}$ and ${\alpha,\beta\in\bF_{q^g}}$, we have that $h(\alpha v+\beta u)=\alpha h(v)+\beta h(u)$. For $V\in\grsmn{q^g}{n/g}{r/g}$ let $H(V)\triangleq\{h(v)\vert v\in V\}$. The set $H(V)$ is clearly a subspace of dimension $r$ over $\bF_q$ in $\bF_{q^n}$. Furthermore, the function $H:\grsmn{q^g}{n/g}{r/g}\to\grsmn{q}{n}{r}$ is injective since $h$ is injective.}



\begin{construction}\cite[Construction~1]{Cyclic}\label{construction:Cd}
For integers $g,n,$ and $r$ such that $0<r<n$ and $g|\gcd(n,r)$, let\[\bC_g\triangleq\{H(V)\vert V\in\grsmn{q^g}{n/g}{r/g}\}.\]
\end{construction}
{Clearly, for $g=1$ Construction~\ref{construction:Cd} is trivial. Thus, we henceforth assume that $g\ge 2$, i.e., $n$ and $r$ have a non-trivial $\gcd$. }The subspace code $\bC_g$ has minimum subspace distance~$2g$, and it may alternatively be defined as direct sums of cyclic shifts of $\bF_{q^g}$ or as the set of all subspace of $\grsmn{q}{n}{r}$ that are subspaces over $\bF_{q^g}$ as well \cite{Cyclic}. Since~$\bC_g$ is the image of an injective function from $\grsmn{q^g}{n/g}{r/g}$ to $\grsmn{q}{n}{r}$, we have the following. 
\begin{corollary}\cite[Corollary~5]{Cyclic}\label{corollary:CdSize}
$|\bC_g|=\qbin{n/g}{r/g}{q^g}$.
\end{corollary}

The subspaces in $\bC_g$ admit a unique subspace polynomial structure, from which the results in this paper follow.

\begin{lemma}\cite[Lemma~14]{Cyclic}\label{lemma:CdPolyStructure}
If $V\in\grsmn{q}{n}{r}$ then $V\in \bC_g$ if and only if $P_V(x)=\sum_{i=0}^{r/g}c_i x^{[gi]}$, where ${c_i \in \bF_{q^n}}, {\forall i \in \{0,\dots,r/g\}}$. 
\end{lemma}

Another concept used in our constructions is the notion of $q$-associates. Two polynomials over $\bF_{q^n}$ of the form $\ell(x)=\sum_{i=0}^d\alpha_i x^i$ and $L(x)=\sum_{i=0}^d\alpha_i x^{q^i}$, are called $q$-associates of each other. For any $g\in\bN$, one can similarly define $q^g$-associativity, where $\ell(x)=\sum_{i=0}^d\alpha_i x^i$, and $L(x)=\sum_{i=0}^d\alpha_i x^{q^{gi}}$ are $q^g$-associates of each other. Linearized polynomials over $\bF_q$ are deeply connected to their $q$-associates as follows.

\begin{lemma}\label{lemma:qAssociates}\cite[Theorem 3.62, p. 116]{FiniteFields}
If $L_1(x)$ and $L(x)$ are linearized polynomials over $\bF_q$ with $q$-associates $\ell_1(x)$ and $\ell(x)$, then $L_1(x)$ divides $L(x)$ if and only if $\ell_1(x)$ divides~$\ell(x)$.
\end{lemma}

\section{Sets of Subspaces Polynomials with Mutual Top Coefficients}\label{section:setOfSSP}

In \cite{SubspacePolynomialsAndLimits} (resp. \cite{BoundsOnListDecodingOfRankMetric}) it was shown that sets of subspace polynomials that agree on many of their top coefficients provide a bound on the list decodability of Reed--Solomon (resp. Gabidulin) codes. By Lemma~\ref{lemma:CdPolyStructure} it is evident that all subspace polynomials of subspaces in $\bC_g$ agree on their topmost $g$ coefficients $(1,0,\ldots,0)$. Using a counting argument we may prove the existence of a subset of $\bC_g$ whose corresponding subspace polynomials agree on a larger number of top coefficients.

\begin{theorem} \label{theorem:MainClaimCounting}
If $g,n,$ and $r$ are integers such that $0<r< n$, $g\vert \gcd(r,n)$, and~$\ell$ is the unique non-negative integer such that $r = n-g(\ell+1)$, then there exists a subset of $\bC_g$ of size at least
\[\frac{\qbin{n/g}{r/g}{q^g}}{q^{n\ell}},\]whose subspace polynomials agree on their topmost $g(\ell+1)$ coefficients.
\end{theorem}

\begin{proof}
Consider the set of all subspace polynomials of subspaces in $\bC_g$ (Construction \ref{construction:Cd}). {Lemma~\ref{lemma:CdPolyStructure} implies that these polynomials have zero coefficients for all monomials $x^{[j]}$ such that $g\nmid j$. } Hence, they may be partitioned into $q^{n\ell}$ subsets according to their $\ell+1$ top coefficients which correspond to monomials whose $q$-degree is divisible by~$g$. According to the pigeonhole principle, there exists a subset of size at least~$\qbin{n/g}{r/g}{q^g}/q^{n\ell}$ whose polynomials agree on their top $g(\ell+1)$ coefficients.
\end{proof}

Notice that for $g=1$, Theorem~\ref{theorem:MainClaimCounting} reduces to the ordinary counting argument employed by \cite{SubspacePolynomialsAndLimits} and \cite{BoundsOnListDecodingOfRankMetric}. In addition, the case where $n-r=g(\ell+1)\ge r$, in which the polynomials in the set agree on \textit{all} coefficients, is also trivial, since it merely implies the existence of a set of size one. Hence, this theorem is applicable only when $r>n/2$. 

The notion of $q^g$-associativity, together with Lemma \ref{lemma:cyclicShiftPoly}, allows us to construct an \textit{explicit} large set of subspace polynomials. It will also be noted that in certain cases, this set of polynomials corresponds to the entire set $\bC_g$. The construction is based on the following lemma.

\begin{lemma}\label{lemma:explicitqAssociate}
If $g,s,$ and $r$ are integers such that $gs|r$ and $n\triangleq r+gs$, then the polynomial $P(x)\triangleq \sum_{i=0}^{n/gs-1}x^{[igs]}$ is a subspace polynomial with respect to $\bF_{q^n}$.
\end{lemma}
\begin{proof}
Since $gs|r$, there exists an integer $\alpha$ such that $gs\alpha=r$, thus $n=gs(\alpha+1)$ and $s|\frac{n}{g}$. It follows that 
\begin{eqnarray*}
\frac{x^{n/g}-1}{x^s-1}=x^{\frac{n}{g}-s}+x^{\frac{n}{g}-2s}+\ldots+1,
\end{eqnarray*}
and hence $(x^{n/g-s}+x^{n/g-2s}+\ldots+1)|(x^{n/g}-1)$. According to Lemma \ref{lemma:qAssociates}, the $q^g$-associates of these polynomials satisfy $\sum_{i=0}^{n/gs-1}x^{[igs]}|(x^{[n]}-x)$, and thus $P$ is a subspace polynomial of an $r$-subspace in $\bF_{q^n}$ by Definition~\ref{definition:subspacePoly}.
\end{proof}

By Lemma \ref{lemma:cyclicShiftPoly} and Lemma \ref{lemma:explicitqAssociate}, we have a large set of subspace polynomials whose coefficients may be given explicitly.

\begin{construction}\label{construction:explicitPolys}
If $g,s,$ and $r$ are integers such that $gs|r$ and $n\triangleq r+gs$, then 
\[
\cZ\triangleq\left\{ \sum_{i=0}^{n/gs-1}\beta ^{[r]-[igs]}x^{[igs]}~\Big|~\beta \in B\right\}
\]
consists of $\frac{q^n-1}{q^{gs}-1}$ subspace polynomials of subspaces in $\grsmn{q}{n}{r}$, where $B$ is any set of nonzero representatives of the orbit of $\bF_{q^{gs}}$.
\end{construction}

\begin{proof}
Since $n=r+gs$ and $gs|r$, it follows that $gs|n$, and thus $\bF_{q^{gs}}$ is a subfield of $\bF_{q^n}$. By Lemma~\ref{lemma:explicitqAssociate}, the polynomial $P_V(x)=\sum_{i=0}^{n/gs-1}x^{[igs]}$ is a subspace polynomial of some $V\in\grsmn{q}{n}{r}$. Let $B$ be any set of representatives of the orbit of $\bF_{q^{gs}}$, that is, a set consisting of a single nonzero element from each subspace in $\{\alpha \bF_{q^{gs}}\vert \alpha\in\bF_{q^n}^*\}$. Since the size of the orbit of $\bF_{q^{gs}}$ is $\frac{q^n-1}{q^{gs}-1}$, and since all subspaces in it intersect trivially \cite[Section III]{ErrorCorrectingCodesInProjectiveSpace}, it follows that $|B|=\frac{q^n-1}{q^{gs}-1}$. By Lemma~\ref{lemma:cyclicShiftPoly}, for all $\beta\in B$ we have that $P_{\beta V}(x)\in \cZ$. We are left to show that if $\beta_1,\beta_2\in B$, then $\beta_1 V\ne \beta_2 V$. 

Assume for contradiction that there exists $\beta_1,\beta_2\in B$ such that $\beta_1 V=\beta_2 V$. It follows that $P_{\beta_1 V}(x)=P_{\beta_2 V}(x)$, and Lemma~\ref{lemma:cyclicShiftPoly} implies that the coefficients of $x$ are equal, that is, $\beta_1^{[n-gs]-1}=\beta_2^{[n-gs]-1}$. Therefore, since every $\alpha\in \bF_{q^n}$ satisfies $\alpha^{q^n}=\alpha$, we have that
\begin{eqnarray*}
\left(\beta_1^{q^{n-gs}-1}\right)^{-q^{gs}}&=&\left(\beta_2^{q^{n-gs}-1}\right)^{-q^{gs}}\\
\beta_1^{q^{gs}-q^n}&=&\beta_2^{q^{gs}-q^n}\\
\beta_1^{q^{gs}-1}&=&\beta_2^{q^{gs}-1}\\
\left(\frac{\beta_1}{\beta_2}\right)^{q^{gs}-1}&=&1.
\end{eqnarray*}
It is widely known (e.g., \cite[Theorem 3.20, p. 91]{FiniteFields}) that the subspace polynomial of $\bF_{q^{gs}}$ is $x^{q^{gs}}-x$, which implies that $\beta_1\beta_2^{-1}\in \bF_{q^{gs}}$, and thus $\beta_1\in \beta_2\bF_{q^{gs}}$. Since $\beta_2\in\beta_2\bF_{q^{gs}}$, it follows that $\beta_1$ and $\beta_2$ belong to the same cyclic shift $\beta_2\bF_{q^{gs}}$, a contradiction.
\end{proof}

Notice that the set $B$ of representatives of $\bF_{q^{gs}}$ (see Construction~\ref{construction:explicitPolys}) may easily be found. For example, if~$\gamma$ is a primitive element of $\bF_{q^n}$, since the set $\{0\}\cup\{\gamma^{i(q^n-1)/(q^{gs}-1)}\}_{i=0}^{q^{gs}-2}$ is $\bF_{q^{gs}}$, it follows that a possible set of representatives of the orbit of $\bF_{q^{gs}}$ is
\[
B\triangleq\left\{\gamma^i~\Big|~ 0\le i\le \frac{q^n-1}{q^{gs}-1}-1\right\}.
\]

\begin{remark}\label{corollary:explicitCd}
For $s=1$, the set $\cZ$ from Construction \ref{construction:explicitPolys} consists of all subspace polynomials of subspaces in~$\bC_g$ (see Construction~\ref{construction:Cd}). This is since the number of cyclic shifts of $\bF_{q^g}$ is $\frac{q^n-1}{q^g-1}$ and the size of $\bC_g$ is $\qbin{n/g}{r/g}{q^g}=\qbin{n/g}{n/g-1}{q^g}=\frac{q^n-1}{q^g-1}$.
\end{remark}

In Section~\ref{section:improved}, we consider subspace polynomials over $\bF_{q^n}$ as polynomials over an extension field $\bF_{q^m}$ of $\bF_{q^n}$. In order to use the above claims over $\bF_{q^m}$, the following formal lemma is required. The proof of this lemma is an immediate corollary of the existence of an injective homomorphism $\phi:\bF_{q^n}\to \bF_{q^m}$.

\begin{lemma}\label{lemma:formalLemma}
Let $P_V(x)=x^{[r]}+\sum_{j=0}^{r-1}v_jx^{[j]}$ and $P_U(x)=x^{[r]}+\sum_{j=0}^{r-1}u_jx^{[j]}$ be two subspace polynomials of subspaces in $\grsmn{q}{n}{r}$, and let $\bF_{q^m}$ be an extension field of $\bF_{q^n}$. If we consider $P_V,P_U$ as polynomials $P_{V'},P_{U'}$ over $\bF_{q^m}$, i.e.,
\begin{eqnarray*}
P_V'(x)=x^{[r]}+\sum_{j=0}^{r-1}v_j'x^{[j]}\\
P_U'(x)=x^{[r]}+\sum_{j=0}^{r-1}u_j'x^{[j]}
\end{eqnarray*}
where the coefficients are in $\bF_{q^m}$, then for all ${j\in\{0,\ldots,r-1\}}$, $v_j=u_j$ if and only if $v_j'=u_j'$. Furthermore, the polynomials $P_{V'},P_{U'}$ are subspace polynomials in $\grsmn{q}{m}{r}$.
\end{lemma}
Notice that generalizing Lemma~\ref{lemma:formalLemma} to the case where $\bF_{q^m}$ is not an extension field of $\bF_{q^n}$, i.e. $U$ and $V$ are subspaces in $\bF_{q^m}$ which are contained in a subspace of dimension $n$, is not clear. However, such a generalization is necessary to use our techniques to bound the list size for any $m \geq n$.

\section{Improved Bounds on List Decodability of Gabidulin Codes}\label{section:improved}
We begin by formally defining Gabidulin codes, which are rank-metric codes that attain a \textit{Singleton}-like bound. Any rank-metric code over $\bF_{q^m}$ of length $n$, minimum rank distance $d$, and size $M$ satisfies $M\le q^{m(n-d+1)}$~\cite{Delsarte,MaximumRankArrayCodes}. 
For a linear rank-metric code of dimension $k$, this bound implies that $d\le n-k+1$. Codes which attain this bound are called \textit{maximum rank distance} (MRD) codes. It can be shown that Gabidulin codes, defined below, are linear MRD codes, attaining ${d=n-k+1}$.

\begin{definition}\cite{TheoryOfCodesWithMaximumRankDistance}
A linear Gabidulin code $\Gab[n,k]$ over $\bF_{q^m}$, length $n \leq m$, and dimension $k\le n$ is the set
\begin{eqnarray*}
\Gab[n,k]\triangleq \left\{\left(P(\alpha_1),\ldots,P(\alpha_n)\right)|\deg_q P< k\right\},
\end{eqnarray*}
where $P$ traverses all $q$-degree restricted linearized polynomials, and $\alpha_1,\ldots,\alpha_n$ are some fixed elements of $\bF_{q^m}$ which are linearly independent over $\bF_q$.
\end{definition}

In \cite{BoundsOnListDecodingOfRankMetric} it was shown that large sets of subspace polynomials that agree on many top coefficients may be used to show the limits of list decoding of Gabidulin codes. For the lack of knowledge about the structure of the coefficients of subspace polynomials, a counting argument was later applied to show the existence of such a set. The resulting bound on list decoding of Gabidulin codes is cited below. In what follows, for $w\in\bF_{q^m}^n$ and $\tau\in\bN$, let $B_\tau(w)\triangleq \{c~\big|\rank(w-c)\le\tau\}$, that is, a ball of radius $\tau$ centered at~$w$.

\begin{theorem}\cite[Theorem 1]{BoundsOnListDecodingOfRankMetric}\label{theorem:countingAntonia}
Consider the code $\Gab[n,k]$ over $\bF_{q^m}$, with $d=n-k+1$. If $\tau<d$, then there exists a word $w\in\bF_{q^m}^n$ such that
\[
\left|\Gab[n,k]\cap B_\tau(w)\right|\ge \frac{\qbin{n}{n-\tau}{q}}{\left(q^m\right)^{n-\tau-k}}.
\]
\end{theorem}
As a result, the following bound is achieved. 
\begin{corollary}\cite[Section III]{BoundsOnListDecodingOfRankMetric}\label{corollary:LDradiusAntonia}
The code $\Gab[n,k]$ over $\bF_{q^m}$, with $d=n-k+1$ cannot be list decoded efficiently for any list decoding radius 
\[
\tau\ge\frac{m+n}{2}-\sqrt{\frac{(m+n)^2}{4}-m(d-\varepsilon)},
\]
for any fixed $0\le\varepsilon<1$.
\end{corollary}

For $n=m$, this bound simplifies to 
\[\tau\ge n-\sqrt{n(n-d+\varepsilon)},\] 
which may be seen as the rank-metric equivalent of the Johnson radius~\cite{AlgorithmicResults}, and for $\varepsilon=0$ it is equal to the Hamming-metric Johnson radius.

By Lemma~\ref{lemma:CdPolyStructure}, in certain cases there exists a large set of subspace polynomials with a unique coefficient structure. Restricting the counting argument used in the proof of Theorem~\ref{theorem:countingAntonia} to the set $\bC_g$ (Theorem~\ref{theorem:MainClaimCounting}) provides a bound which may outperform Corollary~\ref{corollary:LDradiusAntonia}. The proof of the following theorem is illustrated in Fig. \ref{figure:proof}, {and its consequences are discussed in the sequel}.

\begin{figure}
\begin{center}
\definecolor{yqqqqq}{rgb}{0.501960784314,0.,0.}
\definecolor{aqaqaq}{rgb}{0.627450980392,0.627450980392,0.627450980392}
\begin{tikzpicture}[line cap=round,line join=round,>=triangle 45,x=1.0cm,y=1.0cm]
\clip(1.7,1.5) rectangle (10.3,9.1);
\draw [color=aqaqaq,fill=aqaqaq,fill opacity=0.25] (6.,5.) circle (0.75cm);
\draw [line width=2.pt,color=yqqqqq,fill=yqqqqq,fill opacity=0.25] (3.5,5.5) circle (0.75cm);
\draw [line width=2.pt,color=yqqqqq,fill=yqqqqq,fill opacity=0.25] (4.,3.5) circle (0.75cm);
\draw [line width=2.pt,color=yqqqqq,fill=yqqqqq,fill opacity=0.25] (6.,2.5) circle (0.75cm);
\draw [line width=2.pt,color=yqqqqq,fill=yqqqqq,fill opacity=0.25] (8.,3.5) circle (0.75cm);
\draw (6,5) node[] {\Large $c_R$};
\draw (8.3,5.3) node[] {\Huge $ \tau$};
\draw (6,2.5) node[] {\Large $c_{R-P_{1,\beta}}$};
\draw (4,3.5) node[] {\Large $c_{R-P_{2,\beta}}$};
\draw (3.5,5.5) node[] {\Large $c_{R-P_{3,\beta}}$};
\draw (8,3.5) node[] {\Large $c_{R-P_{t,\beta}}$};
\draw(6.,5.) circle (3.47030573843cm);
\draw [->] (6.75,5) -- (9.48,5);
\draw [->] (9.48,5) -- (6.75,5);
\draw [shift={(6.,5.)},dash pattern=on 2pt off 2pt]  plot[domain=-0.303722810845:2.56875222881,variable=\t]({1.*2.68945288687*cos(\t r)+0.*2.68945288687*sin(\t r)},{0.*2.68945288687*cos(\t r)+1.*2.68945288687*sin(\t r)});
\end{tikzpicture}
\caption{An illustration of the proof of Theorem~\ref{theorem:GabidulinCounting}. The proof of Theorem~\ref{theorem:GabidulinExplicit} is~similar.
The ball around $c_R$ of radius $\tau$ contains the words $c_{R-P_{i,\beta}}$ for $P_{i,\beta} \in \mathcal{P_\beta}$, where $|\mathcal{P}_\beta| = \qbin{n/g}{(n-\tau)/g}{q^g}\Big/q^{n\ell}$.
}\label{figure:proof}
\end{center}
\end{figure}
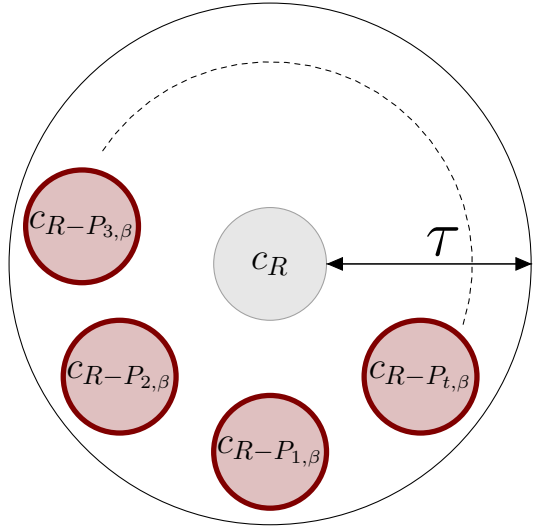

\begin{theorem}\label{theorem:GabidulinCounting}
	For integers $k\le n\le m$ such that $n$ divides $m$, let $\Gab[n,k]$ be a {linear} Gabidulin code over $\bF_{q^m}$, with $d=n-k+1$ and evaluation points $\alpha_1,\ldots,\alpha_n \in\beta\bF_{q^n}$ for some~$\beta\in\bF_{q^m}^*$. Let~$\tau,g$ be integers such that $\floor{\frac{d-1}{2}}+1\le\tau\le d-1$, $g\ge 2$, and $g|\gcd(n-\tau,n)$. If $\ell$ is the unique integer such that $n=n-\tau+g(\ell+1)$ (and thus, $\tau=g(\ell+1)$), then there exists a word $c_R\in\bF_{q^m}^n\setminus \Gab[n,k]$ such that 
\begin{align}\label{equation:countingBound}
\left|\Gab[n,k]\cap B_\tau(c_R)\right|\ge \frac{\qbin{n/g}{(n-\tau)/g}{q^g}}{q^{n\ell}}.
\end{align}
\end{theorem}
\begin{proof}
	According to Theorem~\ref{theorem:MainClaimCounting}, there exists a set $\cP$ of $\qbin{n/g}{(n-\tau)/g}{q^g}/q^{n\ell}$ subspace polynomials of subspaces in~$\grsmn{q}{n}{n-\tau}$, that agree on their topmost $\tau ={g(\ell+1)}$ coefficients. The coefficients of these polynomials are in the field $\bF_{q^n}$. Since $n|m$, we have that $\bF_{q^n}$ is a subfield of~$\bF_{q^m}$, and thus these coefficients may be considered as elements of~$\bF_{q^m}$. Recall that according to Lemma~\ref{lemma:formalLemma}, these polynomials agree on their topmost $\tau$ coefficients also when considered as polynomials over~$\bF_{q^m}$. 
	
	Further, let~$\{V_P\}_{P\in\cP}\subseteq \grsmn{q}{n}{n-\tau}$ be the subspaces which correspond to the subspace polynomials in~$\cP$. For~$P\in\cP$, 
	let~$P_\beta$ be the subspace polynomial of~$\beta V_P$, and let~$\cP_\beta\triangleq\{P_\beta\}_{P\in\cP}$. According to Lemma~\ref{lemma:cyclicShiftPoly}, and according to the properties of~$\cP$, it follows that the polynomials in~$\cP_\beta$ agree on their topmost~$\tau$ coefficients. Since multiplication by~$\beta$ is an injection, it also follows that~$|\cP|=|\cP_\beta|$.
	
	Let $R$ be any linearized polynomial over $\bF_{q^m}$~of $q$-degree $n-\tau$ that has the mutual top coefficients of~$\cP_\beta$, and let $c_R\in\bF_{q^m}^n$ be the word resulting from the evaluation of~$R$ at $\alpha_1,\ldots,\alpha_n$. Similarly, for $P_\beta\in\cP_\beta$ let $c_{R-P_\beta}\in\bF_{q^m}^n$ be the word corresponding to the evaluation of $R-P_\beta$ at $\alpha_1,\ldots,\alpha_n$. 
	
	Since $\deg_q (R-P_\beta)\le n-\tau-g(\ell+1)$ and $\tau =g(\ell+1)> \frac{d-1}{2}=\frac{n-k}{2}$ it follows that $2\tau =\tau+g(\ell+1)> n-k$, and~hence,
	\begin{align*}
	k>  n-\tau-g(\ell+1)\ge\deg_q (R-P_\beta).
	\end{align*}
	
	Therefore, the word $c_{R-P_\beta}$ is a codeword of $\Gab[n,k]$ for all $P_\beta\in\cP_\beta$. In addition, since $\tau\le d-1$ it follows that $\deg_q R=n-\tau\ge n-d+1=k$, and hence $c_R\notin \Gab[n,k]$.
	
	Since every linearized polynomial can be viewed as an $\bF_q$-linear mapping (see Section~I), it follows that for every $P_\beta\in\cP_\beta$,
	\begin{align*}
	\rank(c_R-c_{R-P_\beta})&=\rank((P_\beta(\alpha_1),\ldots,P_\beta(\alpha_n)))\\
	&= \dim\Span{P_\beta(\alpha_1),\ldots,P_\beta(\alpha_n)}\\
	&= \dim P_\beta\left(\Span{\alpha_1,\ldots,\alpha_n}\right)\\
	&= \dim P_\beta(\beta\bF_{q^n}),
	\end{align*}
	where the last equality follows from the fact that $\alpha_1,\ldots,\alpha_n$ are~$n$ linearly independent elements in~$\beta\bF_{q^n}$, a subspace of dimension~$n$.
	Since~$P_\beta$ is a subspace polynomial of~$\beta V_P$, which is a subspace of dimension~$n-\tau$ that is contained in~$\beta\bF_{q^n}$, it follows that $\dim P_\beta(\beta\bF_{q^n})=\tau$.
	Thus, the set $\{c_{R-P_\beta}\}_{P_\beta\in\cP_\beta}\subseteq \Gab[n,k]$ is a set of size $\qbin{n/g}{(n-\tau)/g}{q^g}/q^{n\ell}$, which is contained in  a ball of radius~$\tau$ around the word~$c_R$.
\end{proof}

Notice that the restriction on the parameter $r$, mentioned after the proof of Theorem~\ref{theorem:MainClaimCounting}, implies the necessary condition $r=n-\tau> n/2$, and hence $\tau < n/2$. However, this limitation becomes trivial when discussing $\tau$ which is approximately the unique decoding radius $d/2$, since $d\le n$.

A simple analysis of \eqref{equation:countingBound} shows that 
\begin{eqnarray*}
\left|\Gab[n,k]\cap B_\tau(c_R) \right|&\ge& \frac{\qbin{n/g}{(n-\tau)/g}{q^g}}{q^{n\ell}}\\
&\ge& \frac{(q^g)^{\frac{n-\tau}{g}(\frac{n}{g}-\frac{n-\tau}{g})}}{q^{n\ell}}\\
&=&q^{(n-\tau)\frac{\tau}{g}-n\ell}=q^{\frac{n\tau}{g}-\frac{\tau^2}{g}-n\ell}\\
&=& q^{n(\ell+1)-g(\ell+1)^2-n\ell}\\
&=&q^{n-g(\ell+1)^2}=q^{n-\tau(\ell+1)},
\end{eqnarray*}
and hence, this bound results in a list of exponential size whenever $g(\ell+1)^2<c\cdot n$ for $c\in(0,1)$, or alternatively, when~$\tau < \frac{cn}{\ell+1}$.

The following examples provide infinite sets of Gabidulin codes, with rates from $\frac{1}{5}$ to $1$, that cannot be list decoded efficiently \textit{at all} according to the bound from Theorem~\ref{theorem:GabidulinCounting}. This result strictly outperforms the bound from Corollary \ref{corollary:LDradiusAntonia}, and provides an answer to an open problem by \cite{Ding}, \cite[Section~6]{GuruswamiExplicit}, and \cite[Section~V]{BoundsOnListDecodingOfRankMetric}, that is, there exist Gabidulin codes that cannot be efficiently list decoded beyond the  unique decoding radius.
\begin{example}\label{example:noLD}
Let $n$ be an integer power of 2, and let ${1\le i\le \log n-2}$. For any integer $m$ such that $n|m$, consider a ${\Gab[n,(1-\frac{1}{2^i})n+2]}$ code over $\bF_{q^m}$ with evaluation points that span~$\beta\bF_{q^n}$ for some~$\beta\in\bF_{q^m}^*$, and let $\tau$ be the smallest possible list decoding radius, that is, 
\begin{eqnarray*}
\tau&\triangleq& \bigfloor{\frac{d-1}{2}}+1=\bigfloor{\frac{\frac{n}{2^i}-2}{2}}+1=\frac{n}{2^{i+1}}.
\end{eqnarray*}
Let $g\triangleq \frac{n}{2^{i+1}}=\tau$, and notice that $g\ge 2$. To see that $g|\gcd(n,n-\tau)$, notice that since $n$ is an integer power of~2, it follows that $\tau|n$, and thus $g|n$. In addition, we have that $\tau(2^{i+1}-1)=n-\tau$, thus $\tau|(n-\tau)$ and $g|(n-\tau)$. Therefore, in Theorem~\ref{theorem:GabidulinCounting} we may choose $g=\frac{n}{2^{i+1}},\ell=0$, and get that there exists a word $c_R\in\bF_{q^m}^n$ with $q^{(1-2^{-i-1})n}$ codewords in a ball of radius $\tau$ around it. Since $\tau$ is larger than the unique decoding radius by one, this code cannot be efficiently list decoded at all. {A detailed comparison between this bound and~\cite{BoundsOnListDecodingOfRankMetric} appears in~\nameref{appendix:B}.}

\end{example}

\begin{example}\label{example:ronny}
Let $g,\alpha_n,$ and $\alpha_\tau$ be positive integers such that $\alpha_n\ge \alpha_\tau^2+1$. For $n=\alpha_ng$, $\tau=\alpha_\tau g$, and any integer~$m$ such that $n\vert m$, consider a $\Gab[n,n-2\tau+1]$ code over~$\bF_{q^m}$ with evaluation points that span~$\beta\bF_{q^n}$ for some~$\beta\in\bF_{q^m}^*$, whose minimum distance is $d=2\tau$, and whose rate is
\begin{eqnarray*}
\frac{n-2\tau+1}{n}=1-\frac{2\alpha_\tau}{\alpha_n}+\frac{1}{n}.
\end{eqnarray*}
According to Theorem~\ref{theorem:GabidulinCounting}, there exists a word $c_R$ having
\begin{eqnarray}\label{equation:RonnysExample}
\frac{\qbin{n/g}{(n-\tau)/g}{q^g}}{q^{n\ell}},
\end{eqnarray}
codewords in radius $\tau$ around it, where $\ell=\tau/g-1=\alpha_\tau-1$. Simplifying this expression, we have that 
\begin{eqnarray*}
\frac{\qbin{n/g}{(n-\tau)/g}{q^g}}{q^{n\ell}}&=&\frac{\qbin{\alpha_n}{\alpha_n-\alpha_\tau}{q^g}}{q^{n(\alpha_\tau-1)}}\\
&\ge&\frac{\left(q^g\right)^{(\alpha_n-\alpha_\tau)\alpha_\tau}}{q^{n(\alpha_\tau-1)}}\\
&=&q^{n-\tau\alpha_\tau}=q^{(\alpha_n-\alpha_\tau^2)g}.
\end{eqnarray*}
If $\alpha_\tau$ and $\alpha_n$ are constants then $g=\Omega(n)$ and
$q^{(\alpha_n-\alpha_\tau^2)g}=q^{\Omega(n)}$,
which implies that the list size is exponential in the code length. Since $\tau < n/2$, as mentioned after Theorem~\ref{theorem:GabidulinCounting}, it follows that $\alpha_n > 2\alpha_{\tau}$, and thus we have the following two interesting families of codes.
\begin{enumerate}
\item For $\alpha_n=3$ and $\alpha_\tau =1$ we have the code $\Gab[3g,g+1]$ over any field $\bF_{q^m}$ such that $3g\vert m$, with evaluation points that span~$\beta\bF_{q^{3g}}$ for some~$\beta\in\bF_{q^m}^*$. The rate of this code is $\frac{1}{3}+\frac{1}{n}$, and its minimum distance is $2g$. For the radius $\tau=g$, there exists a word $c_R$ with at least $q^{2g}=q^{\Omega(n)}$ codewords around it, and hence this code cannot be list decoded efficiently at all. 
  
\item For $\alpha_n=5$ and $\alpha_\tau =2$ we have the code $\Gab[5g,g+1]$ over any field $\bF_{q^m}$ such that $5g\vert m$, with evaluation points that span~$\beta\bF_{q^{5g}}$ for some~$\beta\in\bF_{q^m}^*$. The rate of this code is $\frac{1}{5}+\frac{1}{n}$, and its minimum distance is $4g$. For the radius $\tau=2g$, there exists a word $c_R$ with at least $q^{g}=q^{\Omega(n)}$ codewords around it, and hence this code cannot be list decoded efficiently at all.
\end{enumerate}
Clearly, this strategy can be used to construct examples of families with larger code rates, but $\frac{1}{5} + \frac{1}{n}$ {is} the smallest one. This may be seen by considering all integers $\alpha_\tau$ and $\alpha_n$ which comply with the above constraints. {That is, for $\alpha_\tau =1$ and $\alpha_n \ge 4$, the rate is at least $\frac{1}{2}+\frac{1}{n}$, for $\alpha_\tau=2$ and $\alpha_n\ge 6$ the rate is at least $\frac{1}{3}+\frac{1}{n}$, and for any $\alpha_\tau \ge 3$ and any $\alpha_n\ge \alpha_\tau^2+1$ the rate is at least $\frac{1}{3}+\frac{1}{n}$.}
\end{example}

In the following, we present a simple algorithmic way of constructing many dense sets of Gabidulin codewords. These sets also show that the corresponding Gabidulin codes cannot be efficiently list decoded beyond the unique decoding radius. In addition, we have that for certain Gabidulin codes, dense sets of codewords abound and may easily be computed explicitly.

\begin{theorem}\label{theorem:GabidulinExplicit}
Let $g,s,n$, and $m$ be integers such that ${g\ge 2}$, $gs|n$, and $n|m$. Let $\Gab[n,n-2gs+1]$ be a linear Gabidulin code over $\bF_{q^m}$, with $d=2gs$ and evaluation points $\alpha_1,\ldots,\alpha_n\in\beta\bF_{q^n}$ for some $\beta\in\bF_{q^m}^*$. If $\tau\triangleq \floor{\frac{d-1}{2}}+1=gs$, then there exists an (explicitly defined) word $c_R\in\bF_{q^m}^n\setminus\Gab[n,n-2gs+1]$ such that
\[
\left|\Gab[n,n-2gs+1]\cap B_\tau(c_R)\right|\ge\frac{q^n-1}{q^{gs}-1}.
\]
In particular, if $R$ is the polynomial whose evaluation in $\alpha_1,\ldots,\alpha_n$ yields $c_R$, then $\frac{q^n-1}{q^{gs}-1}$ of the codewords in $B_\tau(c_R)$ are given by the evaluations of $\{R-P_\beta\}_{P_\beta\in\cZ_\beta}$ in $\alpha_1,\ldots,\alpha_n$, where $\cZ_\beta$ is the set of subspace polynomials which result from shifting~$\cZ$ (Construction~\ref{construction:explicitPolys}) by~$\beta$.
\end{theorem}

\begin{proof}
	Since $gs|n-gs$, by setting $r=n-gs$ it follows from Construction~\ref{construction:explicitPolys} that the set $\cZ$ is a set of subspace polynomials of subspaces in~$\grsmn{q}{n}{n-gs}$, whose size is $\frac{q^n-1}{q^{gs}-1}$. Since $n|m$, we have that $\bF_{q^n}$ is a subfield of $\bF_{q^m}$, and therefore the polynomials in $\cZ$ may be considered as polynomials over~$\bF_{q^m}$ as well. According to Construction~\ref{construction:explicitPolys} and Lemma~\ref{lemma:formalLemma}, the polynomials in $\cZ$ agree on their topmost $gs$ coefficients $(1,0,\ldots,0)$, even when considered as polynomials over $\bF_{q^m}$. Similar to the proof of Theorem~\ref{theorem:GabidulinCounting}, let $\{V_P\}_{P\in\cZ}$ be the set of subspaces in~$\grsmn{q}{n}{n-gs}$ which corresponds to the subspace polyonmials in~$\cZ$, let $P_\beta$ denote the subspace polynomial of~$\beta V_P$, and let $\cZ_{\beta}\triangleq \{P_\beta\}_{P\in\cZ}$. Clearly, we have that $|\cZ|=|\cZ_\beta|$, and by Lemma~\ref{lemma:cyclicShiftPoly} it follows that the subspace polynomials in~$\cZ_\beta$ agree of their topmost $gs$ coefficients~$(1,0,\ldots,0)$.
	
	Let $R$ be any linearized polynomial of $q$-degree $n-gs$ whose top $gs$ coefficients are 
	$(1,0,\ldots,0)$, and let $c_R\in\bF_{q^m}^n$ be the word resulting from the evaluation of $R$ at $\alpha_1,\ldots,\alpha_n$. For each $P_\beta\in\cZ_\beta$ let $c_{R-P_\beta}\in\bF_{q^m}^n$ be the word corresponding to the evaluation of $R-P_\beta$ at $\alpha_1,\ldots,\alpha_n$. For all $P_\beta\in\cZ_\beta$ we have that $\deg_q (R-P_\beta) \le n-2gs<n-2gs+1$, and thus $c_{R-P_\beta}\in \Gab[n,n-2gs+1]$. In addition, $\deg_q R = n-gs$, and thus $c_R\notin \Gab[n,n-2gs+1]$.
	
	As in the proof of Theorem~\ref{theorem:GabidulinCounting}, for all $P_\beta\in\cZ_\beta$ we have that $\rank(c_R-c_{R-P_\beta})= \dim P_\beta(\beta \bF_{q^n})=gs$. Therefore, the set $\{c_{R-P_\beta}\}_{P_\beta\in\cZ_\beta}$ is a set of $\frac{q^n-1}{q^{gs}-1}$ codewords in ${\Gab[n,n-2gs+1]}$, all of which are of distance at most $\tau=gs$ from $c_R$.
%
%
\end{proof}

Notice that each code in the family of codes mentioned in Theorem~\ref{theorem:GabidulinExplicit} satisfies $d=2gs$, and hence the unique decoding radius is $\floor{ \frac{d-1}{2}}=gs-1$. Furthermore, since $gs|n$, it follows that $gs\le \frac{n}{2}$, and thus the word $c_R$ has $\Omega(q^{n/2})$ codewords in a ball of radius $\tau=\floor{\frac{d-1}{2}}+1$ around it. Hence, this family of Gabidulin codes cannot be list decoded efficiently \textit{at all}.

It is an interesting question if our results can be used to derive a lower bound on the number of words that have an exponentially-sized list of codewords around themselves. If it can be proved that there are just a few just words, we might be able to remove a few codewords of the Gabidulin code to obtain a list decodable code of slightly smaller rate. The code constructed in~\cite{GuruswamiExplicit} seems to be such a list decodable code.

Further, for \emph{folded} Gabidulin codes such a subcode might be easy to find.
The results from~\cite{BartzSidorenko} show that the \emph{average} list size of folded Gabidulin codes is quite small, indicating that there are only a few words with an exponentially-sized list around them.

Finally, the results in this section can be used to prove bounds for punctured Gabidulin codes, which are obtained by removing coordinates from the original code. 
Puncturing a $\Gab[n,k]$ code by $s<n-k+1$ positions yields a $\Gab[n-s,k]$ code.
We can therefore provide lower bounds on list decoding of Gabidulin codes where \emph{$n$ does not divide $m$}.

\begin{lemma}\label{lemma:puncturingBound}
	Let $\cC$ be a $\Gab[n,k]$ code over $\bF_{q^m}$ with minimum distance $d\triangleq n-k+1$, let $s$ be an integer such that $s<d$, and let $\cC_s$ be a $\Gab[n-s,k]$ code which results from $\cC$ by~$s$ puncturing operations, whose minimum distance is ${d'\triangleq n-s-k+1}$. If $\cC$ cannot be list decoded efficiently at all, i.e., there exists a word $w\in\bF_{q^m}^n$ such that
	\[
	\left|\cC\cap B_\tau(w)\right|\ge q^{\Omega(n)}
	\]
	where $\tau\triangleq \floor{\frac{d-1}{2}}+1$, then $\cC_s$ cannot be list decoded efficiently for any radius at least $\tau'+s'$, where $\tau'=\floor{\frac{d'-1}{2}}+1$,~and
	\begin{enumerate}
		\item If $s$ is even, then $s'=\frac{s}{2}$.
		\item If $s$ is odd and $n-k$ is even, then $s'=\frac{s}{2}+\frac{1}{2}$.
		\item If $s$ and $n-k$ are both odd, then $s'=\frac{s}{2}-\frac{1}{2}$.

	\end{enumerate}
\end{lemma}

\begin{proof}
Since puncturing may only reduce the distance between any two given words, and since any two codewords in $\cC$ cannot coincide by puncturing $s<d$ coordinates, it follows that 
\[
	\left|\cC_s\cap B_\tau(w')\right|\ge q^{\Omega(n)},
\]
where $w'\in\bF_{q^m}^{n-s}$ is the result of puncturing $w$. Hence, $\cC_s$ cannot be list decoded efficiently beyond the radius $\tau$. Table~\ref{table} presents the values of $\tau$ as a function of $\tau'$ and $s$, from which the claim follows. 
\end{proof}

\begin{table*}
\begin{center}
\begin{tabular}{ |l|c|c||c|}

	\cline{2-4}
	\multicolumn{1}{c}{} ~ & \multicolumn{1}{|c|}{\multirow{2}{*}{$\tau=\floor{\frac{n-k}{2}}+1$ }}& \multirow{2}{*}{$\tau'=\floor{\frac{n-k-s}{2}}+1$} & \multirow{2}{*}{Resulting radius} \\
	\multicolumn{1}{c}{} ~ &\multicolumn{1}{|c|}{~}&~&~\\ \hhline{-===} 
	
	\multirow{2}{*}{$n-k$ and $s$ are both even.} & \multirow{2}{*}{$\frac{n-k}{2}+1$} & \multirow{2}{*}{$\frac{n-k}{2}-\frac{s}{2}+1$} & \multirow{2}{*}{$\tau=\tau'+\frac{s}{2}$}\\
	~&~&~&~\\\hline
	
	\multirow{2}{*}{$n-k$ is odd and $s$ is even.} & \multirow{2}{*}{$\frac{n-k-1}{2}+1$} & \multirow{2}{*}{$\frac{n-k-1}{2}-\frac{s}{2}+1$} & \multirow{2}{*}{$\tau=\tau'+\frac{s}{2}$}\\
	~&~&~&~\\\hline
	
	\multirow{2}{*}{$n-k$ is even and $s$ is odd.} & \multirow{2}{*}{$\frac{n-k}{2}+1$} & \multirow{2}{*}{$\frac{n-k}{2}-\frac{s+1}{2}+1$} & \multirow{2}{*}{$\tau=\tau'+\frac{s}{2}+\frac{1}{2}$}\\
	~&~&~&~\\\hline
	
	\multirow{2}{*}{$n-k$ and $s$ are both odd.} & \multirow{2}{*}{$\frac{n-k-1}{2}+1$} &\multirow{2}{*}{$\frac{n-k-1}{2}-\frac{s-1}{2}+1$} & \multirow{2}{*}{$\tau=\tau'+\frac{s}{2}-\frac{1}{2}$}\\
	~&~&~&~\\\hline
		
\end{tabular}
\end{center}
\caption{The resulting radius in Lemma~\ref{lemma:puncturingBound}. If $\Gab[n,k]$ cannot be list decoded efficiently for the radius $\tau$, then the punctured code $\Gab[n-s,k]$, $s<n-k+1$, cannot be list decoded efficiently for this radius as well. The rightmost column provides $\tau$ as a function of $\tau'$ and $s$, where the unique decoding radius of $\Gab[n-s,k]$ is $\tau'-1$. The given values for $\tau'$ are simple calculations which follow from $n-k-s$ being either even or odd.}\label{table}
\end{table*}

%

Since the addition to the unique decoding radius $\tau'$ of $\Gab[n-s,k]$ in Lemma~\ref{lemma:puncturingBound} is usually nonzero, it is not clear if those punctured codes indeed cannot be list decoded efficiently at \textit{any} radius. However, for the special case where $s=1$ and $n-k$ is odd, we obtain the following corollary.

\begin{corollary} \label{corollary:onePuncturing}
	For integers $0<k<n$ such that $n-k$ is odd, if $\Gab[n,k]$ cannot be list decoded efficiently at all, i.e., there exist a word $w\in\bF_{q^m}^n$ such that
	\[
	\left|\cC\cap B_\tau(w)\right|\ge q^{\Omega(n)}
	\]
	where $\tau\triangleq \floor{\frac{d-1}{2}}+1$, then the punctured code $\Gab[n-1,k]$ cannot be list decoded efficiently at all.
\end{corollary}

Although Corollary~\ref{corollary:onePuncturing} does not provide a drastic improvement in the variety of codes to which our bounds apply, it does imply the important observation that the divisibility constraints between $n$ and $m$ in Theorem~\ref{theorem:GabidulinCounting} and Theorem~\ref{theorem:GabidulinExplicit} \emph{are not necessary}. In addition, one may obtain infinite examples of Corollary~\ref{corollary:onePuncturing} by puncturing either of the codes $\Gab[3g,g+1]$ and $\Gab[5g,g+1]$ from Example~\ref{example:ronny}, and thus obtain $\Gab[3g-1,g+1]$ and $\Gab[5g-1,g+1]$ codes that cannot be list decoded efficiently at all.

\color{black}
\section{Bounds for Constant-Dimension Subspace Codes}\label{section:subspaceCodes}
In this section, we state new bounds on list decoding \emph{lifted Gabidulin codes} (see~\cite{ARankMetricApproach}), which are a class of almost-optimal constant dimension subspace codes. 
Lifted Gabidulin codes are of special interest since, in contrast to many other subspace code constructions, they can be efficiently decoded (see~\cite{ARankMetricApproach}) while only losing a relatively small number of codewords compared to other subspace code constructions.
These bounds are a direct consequence of our bounds for list decoding Gabidulin codes (Theorem~\ref{theorem:GabidulinCounting} and Theorem~\ref{theorem:GabidulinExplicit}). 

Throughout this section, the quadruple $(n,M_s,d_s,r)_q$ denotes a constant dimension subspace code in the Grassmannian $\grsmn{q}{n}{r}$ of cardinality $M_s$ and minimum subspace distance $d_s$. Further, $\Span{A}$ denotes the subspace spanned by the rows of a matrix $A$. The \emph{lifting} is a map which is applied to a single matrix or a set of matrices and is defined as follows.
\begin{definition} \label{def:lifting_matrixcode}
Consider the mapping
\begin{align*}
\mathcal{I}: \quad \bF_q^{n \times m} &\rightarrow \grsmn{q}{n}{n+m}\\
X &\mapsto \Span{[\I_n \ X]},
\end{align*}
where $\I_n$ denotes the $n \times n $ identity matrix.
The subspace $\liftmap{\X} = \Span{[\I_n \ \X]}$ is called {lifting} of the matrix $\X$. If we apply this map on all codewords of a code $\mathcal{C}$ (in matrix representation), then the subspace code $\liftmap{\mathcal{C}}$ is called lifting of the code ${C}$.
\end{definition}

The properties of a lifted code were studied by Silva, Kschischang and K\"otter and are summarized in the following two lemmas.

\begin{lemma}\cite{ARankMetricApproach}\label{lem:rankandsubdistance}
Let $\X, \Y \in \Fq^{n \times m}$ and let $\liftmap{\X}, \liftmap{\Y} \in \grsmn{q}{n+m}{n}$ be as in Definition~\ref{def:lifting_matrixcode}. Then,
\begin{equation*}
d_s\left(\liftmap{\X}, \liftmap{\Y}\right) = 2 \cdot d_R(\X, \Y ).
\end{equation*}
\end{lemma}
\begin{proof} 
\begin{align*}
d_s\left(\liftmap{\X},\liftmap{\Y}\right)  
& = 2 \dim \left(\liftmap{\X}+\liftmap{\Y}\right)\\
&\phantom{=}-\dim\left(\liftmap{\X}\right)-\dim\left(\liftmap{\Y}\right)\\
& = 2  \rank \begin{pmatrix}
\I_n & \X\\
\I_n & \Y\\
\end{pmatrix} - 2n
\\
&= 2  \rank \begin{pmatrix}
\I_n & \X\\
\mathbf{0} & \Y-\X\\
\end{pmatrix} - 2n\\
& = 2 \left[\rank(\I_n)+\rank(\X-\Y)\right]-2n 
\\&= 2 \rank(\X-\Y)= 2 d_R(\X,\Y).
\end{align*} 
\end{proof}
The following lemma directly follows from Lemma~\ref{lem:rankandsubdistance}.
\begin{lemma}\cite{ARankMetricApproach}\label{lem:subspacecode_lifted_rank}
Let $\mathcal{C}$ be a rank-metric code over $\bF_{q^m}$ of length $n\leq m$, minimum rank distance $d_R$ and cardinality~$M_R$, whose codewords are represented as $m \times n$ matrices over $\bF_q$.
Then, the lifting of the transposed codewords, i.e.,
\begin{equation*}
\liftmap{\mathcal{C}^T}\triangleq\Big\{\liftmap{\C^T} = \Span{[\I_n \ \C^T]} ~\Big|~ \C \in \mathcal{C}\Big\}
\end{equation*}
is an $(n+m,M_s=M_R,d_s=2d_R,n)_q$ constant dimension subspace code.
\end{lemma}
Hence, the lifting of the transpose of a $\Gab[n,k]$ code over $\mathbb{F}_{q^{m}}$ with $n\leq m$, minimum rank distance ${d=n-k+1}$ and cardinality $M_R=q^{mk}$ results in an $(n+m,q^{mk},2d,n)_q$ constant dimension subspace code in the Grassmannian $\grsmn{q}{n+m}{m}$.

So far, the only known bound to list decoding subspace codes was given in~\cite{RosenthalSubspace} and is based on the results for Gabidulin codes from~\cite{BoundsOnListDecodingOfRankMetric}. The following theorem summarizes the result from~\cite{RosenthalSubspace}.
\begin{theorem}\cite[Theorem~37]{RosenthalSubspace}
	Let $\mathcal{C}$ be a linear $\Gab[n,k]$ Gabidulin code over $\bF_{q^{m}}$ of length $n \leq  m$, $d=n-k+1$, evaluation points $\alpha_1,\ldots,\alpha_{n} \in\bF_{q^{m}}$, and let $\tau$ be an integer such that $\lfloor\tau/2 \rfloor < d$. Denote by $\liftmap{\mathcal{C}^T}$ the $(n+m,q^{mk},2d,n)_q$ subspace code from the lifting of the code $\mathcal{C}$ as in Definition~\ref{def:lifting_matrixcode}.
	Then, there is a subspace $\Span{R}$ such that
	\begin{equation*}
		\left|\liftmap{\mathcal{C}^T}\cap B^s_\tau(\Span{R})\right|
		\geq \frac{\qbin{n}{\lfloor\tau/2\rfloor}{q}}{q^{m(n-k-\lfloor\tau/2\rfloor)}}.
	\end{equation*}
\end{theorem}

Let $B^s_{\tau}(\Span{W}) \triangleq \{\Span{V}~\big|\; d_s(\Span{W}, \Span{V})\leq\tau\}$ denote a ball of radius $\tau$ centered at $\Span{W}$ in the subspace distance.
With Lemma~\ref{lem:rankandsubdistance}, we obtain the following relation between a rank-metric code $\mathcal{C}$ and its lifted subspace code $\liftmap{\mathcal{C}^T}$:
\begin{equation}\label{eq:relation_Balls_rank_subspace}
\left|\;\mathcal{C} \cap B_\tau(c_R)\right| \leq \left|\liftmap{\mathcal{C}^T}\cap B^s_{2\tau}(\liftmap{c_R^T})\right|.
\end{equation}
This relation and Theorem~\ref{theorem:GabidulinCounting} provide the following theorem on the list size of lifted Gabidulin codes.


\begin{theorem}\label{theorem:LiftedGabidulinCounting_gen}
Let $\mathcal{C}$ be a linear $\Gab[n,k]$ Gabidulin code over~$\bF_{q^{m}}$ with length $n \mid m$, $d=n-k+1$, and evaluation points $\alpha_1,\ldots,\alpha_n \in\beta\bF_{q^n}$ for some~$\beta\in\bF_{q^m}^*$. Let $\tau,g$ be integers such that $\floor{\frac{d-1}{2}}+1\le\lfloor\frac{\tau}{2}\rfloor\le d-1$, $g\ge 2$, and $g|\gcd(n-\lfloor\frac{\tau}{2}\rfloor,n)$. Let~$\ell$ be the unique integer such that $n=n-\lfloor\frac{\tau}{2}\rfloor+g(\ell+1)$ (and thus, $\lfloor\frac{\tau}{2}\rfloor=g(\ell+1)$) and denote by $\liftmap{\mathcal{C}^T}$ the $(n+m,q^{mk},2d,n)_q$ subspace code from the lifting of the code $\mathcal{C}$ as in Definition~\ref{def:lifting_matrixcode}.

Then there exists a subspace $\liftmap{c_R^T}\in \grsmn{q}{n+m}{n}$, where $c_R\in\bF_{q^{m}}^n\setminus \Gab[n,k]$ such that 
\begin{align*}
\left|\liftmap{\mathcal{C}^T}\cap B^s_\tau(\liftmap{c_R^T})\right|
&\geq \frac{\qbin{n/g}{(n-\lfloor\tau/2\rfloor)/g}{q^g}}{q^{n\ell}}\nonumber\\
&\geq q^{n-\lfloor\tau/2\rfloor(\ell+1)}.\label{equation:LiftedcountingBound_gen}
\end{align*}
\end{theorem}
\begin{proof}
The statement follows from~\eqref{eq:relation_Balls_rank_subspace} and Theorem~\ref{theorem:GabidulinCounting}. 
The floor operation for $\lfloor\tau/2\rfloor$ is necessary since the subspace distance is an even number, see explanation of the proof of \cite[Theorem~37]{RosenthalSubspace}.
\end{proof}
Thus, this bound results in a list of exponential size for even~$\tau$ when $\tau < \frac{2cn}{\ell+1}$ and for odd $\tau$ when $\tau < \frac{2cn}{\ell+1}+1$ for~${c\in(0,1)}$, which results for many cases in a better bound than the one from \cite[Theorem~37]{RosenthalSubspace}. Similarly, from Theorem~\ref{theorem:GabidulinExplicit}, we obtain the following theorem.
\begin{theorem}\label{theorem:GabidulinExplicit_Subspace}
Let $g,s,n$, and $m$ be integers such that $g\ge 2$, $gs|n$, and $n|m$.  Let $\mathcal{C}$ be a linear $\Gab[n,n-2gs+1]$ Gabidulin code over $\bF_{q^m}$ , with $d=2gs$ and evaluation points $\alpha_1,\ldots,\alpha_n \in\beta\bF_{q^n}$ for some~$\beta\in\bF_{q^m}^*$.
Denote by $\liftmap{\mathcal{C}^T}$ the ${(n+m,q^{m(n-2gs+1)},2d,n)_q}$ subspace code from the lifting of the code $\mathcal{C}$ as in Definition~\ref{def:lifting_matrixcode}.

If $\lfloor\frac{\tau}{2}\rfloor\triangleq \floor{\frac{d-1}{2}}+1=gs$, then there exists an (explicitly defined) subspace $\liftmap{c_R^T}\in \grsmn{q}{n+m}{n}$, where \[c_R\in\bF_{q^m}^n\setminus\Gab[n,n-2gs+1],\]such that
\[
\left|\liftmap{\mathcal{C}^T}\cap B^s_\tau(\liftmap{c_R^T})\right|\ge\frac{q^n-1}{q^{gs}-1} = \frac{q^n-1}{q^{\lfloor{\tau}/{2}\rfloor}-1}.
\]
\end{theorem}
The explicitly defined subspace follows directly from lifting the matrix representation of the explicit word of Theorem~\ref{theorem:GabidulinExplicit}. In~\cite{BoundsOnListDecodingOfRankMetric}, a non-linear rank-metric code was presented which cannot be list decoded efficiently at all. The lifting of this code obviously results in a subspace code with the same restrictions to list decoding as lifted Gabidulin codes.
However, lifted Gabidulin codes are of special interest for network coding and therefore, we have analyzed their list decoding capability in this section.

\section{Conclusions and Future Work}\label{section:conclusions}
We have improved the worst-case bound on the list decodability of Gabidulin codes in many cases. 
This was shown by using the structure of the subspace polynomials of a subset of $\grsmn{q}{n}{r}$ for $n$ and $r$ that have a non-trivial $\gcd$. In addition, we have presented such subspace polynomials explicitly, using the notions of cyclic shifts and $q$-associativity. Both of these results outperform the counting argument applied in \cite{BoundsOnListDecodingOfRankMetric}, and provide examples of infinite families of Gabidulin codes that cannot be list decoded efficiently beyond the unique decoding radius. This resolves an open question by~\cite{Ding,GuruswamiExplicit}, and~\cite{BoundsOnListDecodingOfRankMetric} and reveals a significant difference between decoding Gabidulin and Reed--Solomon codes despite their similar code definitions.

The work of \cite{BoundsOnListDecodingOfRankMetric} ruled out the existence of an efficient algorithm for list decoding of Gabidulin codes beyond the Johnson radius. Our work rules out the existence of an efficient list decoding algorithm that applies for any Gabidulin code and any radius beyond half the minimum distance. However, this certainly does not rule out the existence of an efficient algorithm for list decoding of very large subcodes of Gabidulin codes or Gabidulin codes with lower rates, since our work requires the code parameters to satisfy some strict number-theoretic and field-theoretic constraints, and our examples have rate at least~$\frac{1}{5}$. 
For example,~\cite{GuruswamiExplicit} provides a subcode of a Gabidulin code which can be list decoded efficiently. 

We have also shown that identical results hold for lifted Gabidulin codes, which are an important class of nearly optimal subspace codes.
Additional discussion about the inapplicability of our techniques to improve the known combinatorial bound on list decoding of Reed--Solomon codes appears in~\nameref{appendix:A}. 

For future research, we would like to have similar bounds on Gabidulin codes in $\bF_{q^m}^n$ where the evaluation points do not necessarily come from a cyclic shift of~$\bF_{q^n}$. This seems to require a rigorous understanding of the connection between the subspace polynomials of a given subspace~$V$ and the subspace $A\cdot V$, where~$A$ is a nonsingular transform. Moreover, we would like to generalize our results for \emph{any} case where $n$ does not necessarily divide $m$, a problem which seems to require generalizing Lemma~\ref{lemma:formalLemma} to the case $n\nmid m$. 
In addition, we would like to derive bounds for Gabidulin codes with rates less than~$\frac{1}{5}$.

\section*{Appendix~A}\label{appendix:A}
In \cite{SubspacePolynomialsAndLimits}, limits for list decoding of Reed--Solomon codes were shown using techniques which highly resemble the ones in \cite{BoundsOnListDecodingOfRankMetric} and in this paper. The interested reader might conjecture that the improvement achieved here (see Theorem~\ref{theorem:GabidulinCounting} and Theorem~\ref{theorem:GabidulinExplicit}) for Gabidulin codes may also be attained for Reed--Solomon codes, for which list decoding related problems were very extensively studied. In what follows we briefly describe why such an improvement cannot be directly attained by our techniques. Adapting these techniques to Reed--Solomon codes remains an intriguing open problem. In the sequel, we briefly describe the methods and results of \cite{SubspacePolynomialsAndLimits}.

Following the notations in \cite{SubspacePolynomialsAndLimits}, a Reed--Solomon code $\RS[q^n,q^u]$ of length $q^n$ and dimension $q^u$ is a subset of $\bF_{q^n}^{q^n}$ such that
\[
\RS[q^n,q^u] \triangleq \left\{ \left(p(\alpha_1),\ldots,p(\alpha_{q^n})\right)~\Bigg|~\begin{matrix}p:\bF_{q^n}\to\bF_{q^n}\mbox{ is a }\\\mbox{polynomial with }\\\deg(p)< q^u\end{matrix} \right\},
\]
where $\{\alpha_i\}_{i=1}^{q^n}$ are \textit{all} elements of $\bF_{q^n}$. Notice that Reed--Solomon codes may be defined as the evaluation of polynomials in any number of elements in the field. However, we consider this definition for convenience. Notice also that any word $w\in \bF_{q^n}^{q^n}$ may be regarded as a polynomial over $\bF_{q^n}$, and any word $c\in \RS[{q^n},{q^u}]$ may be regarded as a polynomial over $\bF_{q^n}$ of \textit{bounded degree}. 

\begin{definition}\cite[Definition 3.3]{SubspacePolynomialsAndLimits}\label{definition:asFamily}
	A family of polynomials $\cP\subseteq \bF_{q^n}[x]$ is said to be an $(a,s)$-family if
	\begin{enumerate}
		\item Each polynomial in $\cP$ has at least $a$ roots in $\bF_{q^n}$.
		\item There is a polynomial $P^*$ such that for all $P\in\cP$, $P^*-P$ has degree at most $s$. We refer to $P^*$ as a \textit{pivot} of the family.
	\end{enumerate}
\end{definition}

\begin{lemma}\cite[Proposition 3.5]{SubspacePolynomialsAndLimits}\label{lemma:FamilyAndBound}
	Let $a,s$ and $\ell$ be positive integers. Then, the following are equivalent.
	\begin{enumerate}
		\item There is a word $w:\bF_{q^n}\to\bF_{q^n}$ and $\ell$ polynomials $P_1,\ldots,P_\ell$ of degree at most $s$ such that for $i=1,2,\ldots,\ell$, $P_i$ and $w$ agree on at least $a$ points of $\bF_{q^n}$.
		\item There is an $(a,s)$-family of size $\ell$ of polynomials, whose pivot is the unique polynomial $P_w$ that agrees with the word $w$ on all elements in $\bF_{q^n}$.
	\end{enumerate}
\end{lemma}

The polynomial $P_w$ corresponds to the ``problematic'' word, that is, the word that has exponentially many codewords in a small radius around it. The polynomials $P_1,\ldots,P_\ell$, having a low degree, are the codewords surrounding $P_w$. It is readily verified that the polynomials $P_1,\ldots,P_\ell$ are inside a ball of small radius centered at $P_w$ if and only if the polynomials $\{P_w-P_i\}_{i=1}^s$ have multiple roots in $\bF_{q^n}$. As subspace polynomials have many roots over $\bF_{q^n}$, they are good candidates for playing the role of the polynomials $\{P_w-P_i\}_{i=1}^s$. This intuition is formalized as follows.

\begin{lemma}\cite{SubspacePolynomialsAndLimits}\label{lemma:TopCoefficients}
	If $S\subseteq \grsmn{q}{n}{r}$ is a set of subspaces whose corresponding subspace polynomials have identical $r-t$ top coefficients for some integer $t<r$, then the set of subspace polynomials of $S$ forms a $(q^r,q^t)$-family.
\end{lemma}
\begin{proof}
	Let $W$ be the set of subspace polynomials of the subspace in the set $S$. Since every polynomial in $W$ is a subspace polynomial, it has exactly $q^r$ roots in $\bF_{q^n}$.  If $P_w$ is the linearized polynomial consisting of the $r-t$ mutual top coefficients of the polynomials in $W$, then $\deg(P_w-P_i)\le q^t$ for all $P_i\in W$.
\end{proof}

In light of Lemma \ref{lemma:FamilyAndBound} and Lemma \ref{lemma:TopCoefficients}, presenting a large family of subspace polynomials that agree on many top coefficients suffices for providing a word that is adjacent to too many Reed--Solomon codewords. Such a family of size $\qbin{n/g}{r/g}{q^g}/q^{n\ell}$ was presented in Theorem~\ref{theorem:MainClaimCounting}, where $g|\gcd(n,r)$ and $\ell=\frac{n-r}{g}-1$. Using the standard bound on the Gaussian coefficient (see Section~\ref{section:introduction}) we have that 
\[
\frac{\qbin{n/g}{r/g}{q^g}}{q^{n\ell}}\le 4q^{\frac{r}{g}(n-r)-n\ell}.
\]
Plugging in the expression for $\ell$ results in an upper bound of~$4q^n$, and hence the size of the family is not more than 4 times the length of the code, which is~$q^n$. In addition, an explicit family can be derived from Construction~\ref{construction:explicitPolys} whose size is not super-polynomial in $n$ either, and hence a super-polynomial list decoding bound is \textit{not} achieved.

Both of these families do provide dense sets that are larger than the ones achieved by a counting argument. Dense sets of Reed--Solomon codewords have applications in hardness of approximating the minimum distance of a linear code \cite{dumer2003hardness} and in constructing error-correcting codes with improved parameters \cite{xing2003nonlinear}. However, the dense sets provided by our results are not nearly large enough for these applications.

\section*{Appendix~B}\label{appendix:B}
The following lemmas shows that the bound from Theorem~\ref{theorem:GabidulinCounting} strictly outperforms the bound implied by Theorem~\ref{theorem:countingAntonia} and Corollary~\ref{corollary:LDradiusAntonia}, given in~\cite{BoundsOnListDecodingOfRankMetric},
when applied over the code in Example~\ref{example:noLD}.

\begin{lemma}\label{lemma:appendixTechnical}
	For any $i\ge 1$,
	\[
	1-\sqrt{\frac{2^i-1}{2^i}}> \frac{1}{2^{i+1}}.
	\]
\end{lemma}
\begin{proof}
	Clearly, $\frac{1}{2^{i+2}}> 0$, and hence,
	\begin{eqnarray*}
		2^i-1+\frac{1}{2^{i+2}}&>& 2^i-1\\
		1-\frac{2}{2^{i+1}}+\frac{1}{2^{2i+2}}&>& \frac{2^i-1}{2^i}\\
		\left( 1-\frac{1}{2^{i+1}}\right)^2&>& \frac{2^i-1}{2^i}\\
		1-\frac{1}{2^{i+1}}&>&\sqrt{ \frac{2^i-1}{2^i}}\\
		1-\sqrt{\frac{2^i-1}{2^i}}&>& \frac{1}{2^{i+1}}
	\end{eqnarray*}
\end{proof}

\begin{lemma}\label{lemma:appendix_comparison}
	For a large enough $n$, the radius $\tau=\frac{n}{2^{i+1}}$, for which the code in Example~\ref{example:noLD} cannot be list decoded efficiently according to Theorem~\ref{theorem:GabidulinCounting}, is strictly smaller than the radius $\tau'$ which is guaranteed by the Corollary~\ref{corollary:LDradiusAntonia}.
\end{lemma}
\begin{proof}
	Inserting $\varepsilon=1$ into the bound of Corollary~\ref{corollary:LDradiusAntonia} provides a stronger bound than Corollary~\ref{corollary:LDradiusAntonia} for any $\varepsilon<1$. 
	Therefore, when our bound outperforms Corollary~\ref{corollary:LDradiusAntonia} with $\varepsilon=1$, our bound also outperforms Corollary~\ref{corollary:LDradiusAntonia} with $\varepsilon<1$.
	
	Since in Example~\ref{example:noLD} we have $d=\frac{n}{2^i}-1$ it follows that
	\begin{eqnarray}
	\nonumber\tau'&\ge& \frac{m+n}{2}-\sqrt{\frac{(m+n)^2}{4}-m(d-1)}\\
	\nonumber&=&\frac{m+n}{2}-\sqrt{\frac{(m+n)^2}{4}-m\left(\frac{n}{2^i}-2\right)}\\
	&=&\left( \frac{m+n}{2}\right)\left( 1-\sqrt{1-\frac{4m(d-1)}{(m+n)^2}}\right).
	\label{equation:appendixBound}
	\end{eqnarray}
	Notice that by Theorem~\ref{theorem:countingAntonia}, the bound of~\cite{BoundsOnListDecodingOfRankMetric} is weaker if ${m>n}$, whereas the bound of Theorem~\ref{theorem:GabidulinCounting} does not depend on $m$. Therefore, it suffices to show that the bound from Theorem~\ref{theorem:GabidulinCounting} outperforms the one from~\cite{BoundsOnListDecodingOfRankMetric} for $m=n$. In this case, (\ref{equation:appendixBound}) simplifies to
	\begin{eqnarray}\label{equation:appendixBound2}
	\tau'\ge n\left(1-\sqrt{1-\frac{1}{2^i}+\frac{2}{n}}\right).
	\end{eqnarray}
	For a large enough $n$ the term $\frac{2}{n}$ may be neglected. Hence, by Lemma~\ref{lemma:appendixTechnical}, (\ref{equation:appendixBound2}) implies that 
	\begin{eqnarray*}
		\tau'\ge n\left(1-\sqrt{\frac{2^i-1}{2^i}}\right)>\frac{n}{2^{i+1}}=\tau.
	\end{eqnarray*}
\end{proof}

\section*{Acknowledgment}
The authors would like to thank Ron M. Roth for bringing up the idea of puncturing Gabidulin codes (Lemma~\ref{lemma:puncturingBound} and Corollary~\ref{corollary:onePuncturing}), and the idea of Example~\ref{example:ronny}. The authors would also like to thank the reviewers whose comments helped to improve the presentation of the paper, and finally, to Pierre Loidreau, for pointing out an error in previous versions of this paper.

\IEEEbiographynophoto {Netanel Raviv} (S'15) received a B.Sc. degree from the department of mathematics and an M.Sc. degree from the department of Computer Science at the Technion--Israel Institute of Technology, Haifa,
Israel, at 2010 and 2013, respectively. He is now a Doctoral student at the department of Computer Science at the Technion. He is an awardee of the IBM Ph.D. fellowship for the academic year of 2015-2016, and the Aharon and Ephraim Katzir study grant for 2015. His research interests include coding for distributed storage systems, algebraic coding theory, network coding, and algebraic structures. 

\IEEEbiographynophoto {Antonia Wachter-Zeh}(S'10--M'14) received a 
B.S. degree in electrical engineering in 2007 from the University of Applied Science Ravensburg, Germany, 
and the M.S. degree in communications technology in 2009 from Ulm University, Germany.
She obtained her Ph.D. degree in 2013 at the Institute of Communications Engineering, University of Ulm, Germany and at the 
Institut de recherche mathématique de Rennes (IRMAR), Université de Rennes 1, Rennes, France.
Currently, she is a postdoctoral researcher at the Technion---Israel Institute of Technology, Haifa, Israel.
Her research interests are coding and information theory and their applications.


\begin{thebibliography}{1}

\bibitem{BartzSidorenko}
H.~Bartz and V.~Sidorenko, ``List and Probabilistic Unique Decoding of Folded Subspace Codes,'' {\em IEEE International Symposium on Information Theory (ISIT)}, pp.~11--15, 2015.

\bibitem{Cyclic}
E.~Ben-Sasson, T.~Etzion, A.~Gabizon, and N.~Raviv, ``Subspace polynomials and cyclic subspace codes,'' {\em IEEE Transactions on Information Theory,} vol.~62, no.~1, pp.~1--9, 2016.
    
\bibitem{AffineDispersers}
E.~Ben-Sasson and S.~Kopparty, ``Affine dispersers from subspace polynomials,'' {\em SIAM Journal on Computing}, vol.~41, no.~4, pp.~880--914, 2012.
  
\bibitem{SubspacePolynomialsAndLimits}
E.~Ben-Sasson, S.~Kopparty, and J.~Radhakrishnan, ``Subspace polynomials and
  limits to list decoding of {R}eed {S}olomon-codes,'' {\em IEEE Transactions on Information Theory}, vol.~56, no.~1, pp.~113--120, 2010.
  
\bibitem{ConstructingHighOrderElements}
Q.~Cheng, S.~Gao, and D.~Wan, ``Constructing high order elements through subspace polynomials,'' {\em Symposium on Discrete Algorithms (SODA)}, pp.~1457--1463, 2012.

\bibitem{Delsarte}
P.~Delsarte, ``Bilinear forms over a finite field, with applications to coding theory,'' {\em Journal of Combinatorial Theory}, Series A, vol. 25, no. 3, pp. 226--241, 1978.

\bibitem{Ding}
Y.~Ding, ``On list decodability of random rank metric codes and subspace codes.'' {\em IEEE Transactions on Information Theory}, vol. 61, no.1, pp.51--59, 2015.‏

\bibitem{dumer2003hardness}
I.~Dumer, D.~Micciancio, and M.~Sudan, ``Hardness of approximating the minimum
  distance of a linear code,'' {\em IEEE Transactions on Information Theory}, vol.~49, no.~1, pp.~22--37, 2003.
  
\bibitem{ErrorCorrectingCodesInProjectiveSpace}
T.~Etzion and A.~Vardy, ``Error-correcting codes in projective space,'' {\em IEEE Transactions on Information Theory}, vol.~57, no.~2, pp.~1165--1173, 2011.
  
\bibitem{TheoryOfCodesWithMaximumRankDistance}
E.~M.~Gabidulin, ``Theory of codes with maximum rank distance,'' {\em Problems of Information Transmission (English translation of Problemy Peredachi Informatsii)}, vol.~21, 1985.
  
\bibitem{Gabidulin1991Ideals}
E.~M. Gabidulin, A.~V. Paramonov, and O.~V. Tretjakov, ``{Ideals over a noncommutative ring and their applications to cryptography},'' \emph{Eurocrypt}, pp.~482-489, 1991.
 
\bibitem{ConstantRankCodesAndTheir}
M.~Gadouleau and Z.~Yan, ``Constant-rank codes and their connection to constant-dimension codes,'' {\em IEEE Transactions on Information Theory},  vol.~56, no.~7, pp.~3207--3216, 2010.

\bibitem{AlgorithmicResults}
V.~Guruswami, {\em Algorithmic results in list decoding}, Now Publishers Inc, 2006.

\bibitem{Guruswami-Sudan}
V.~Guruswami and M.~Sudan, ``{Improved decoding of Reed--Solomon and
  Algebraic--Geometry Codes},'' \emph{IEEE Transactions on Information Theory}, vol.~45, no.~6, pp.~1757--1767, Sep. 1999.

\bibitem{GuruswamiExplicit}
V.~Guruswami and C.~Wang, ``Evading subspaces over large fields and explicit list decodable rank-metric codes,'' {\em APPROX-RANDOM}, pp. 748–-761, 2014.‏‏

\bibitem{GuruswamiExplicit2}
V.~Guruswami, S.~Narayanan, and C.~Wang, ``List decoding subspace codes from insertions and deletions,'', {\em Innovations in Theoretical Computer Science (ITCS)}, pp.~183--189, 2012.

\bibitem{KandK}
R.~K\"{o}tter and F.R.~Kschischang, ``Coding for errors and erasures in random network coding,'' {\em IEEE Transactions on Information Theory}, vol.~54, no.~8, pp.~3579--3591, 2008.‏
  
\bibitem{FiniteFields}
R.~Lidl and H.~Niederreiter, {\em {Finite Fields}}, vol.~20 of {\em Encyclopedia of Mathematics and Its Applications}. \newblock Cambridge University Press, 1997.

\bibitem{Loidreau2010Designing}
P.~Loidreau, ``{Designing a rank metric based McEliece cryptosystem},'' \emph{Post-Quantum Cryptography}, pp. 142--152, 2010.
  
\bibitem{Lusina2003Maximum}
  P.~Lusina, E.~M.~Gabidulin, and M.~Bossert, ``{Maximum rank distance codes as space-time codes},'' \emph{IEEE Transactions on Information Theory}, vol.~49, no.~10, pp.~2757--2760, 2003.
    
\bibitem{Lu2004Generalized}
H.~F. Lu and P.~V. Kumar, ``{Generalized unified construction of space-time codes with optimal rate-diversity tradeoff},'' \emph{IEEE International Symposium on Information Theory (ISIT)}, p.~95--99, 2004.

\bibitem{LD1}
H.~Mahdavifar and A.~Vardy, ``Algebraic list decoding of subspace codes," {\em IEEE Transactions on Information Theory}, vol. 59, no. 12, pp. 7814--7828, 2013.

\bibitem{LD2}
H.~Mahdavifar and A.~Vardy, ``Algebraic list decoding of subspace codes with multiplicities,'' {\em 49th Annual Allerton Conference on Communication, Control, and Computing}, pp. 1430--1437, 2011.

\bibitem{OnASpecialClassOfPolynomials}
\O.~Ore, ``On a special class of polynomials,'' {\em Transactions of the American Mathematical Society}, vol.~35, no.~3, pp.~559--584, 1933.
  
\bibitem{RavivWachterzeh-ISITListDecoding}
N.~Raviv and A. Wachter-Zeh, ``Some Gabidulin codes cannot be list decoded efficiently at any radius,'' \emph{IEEE International Symposium on Information Theory (ISIT)}, pp.6--10, 2015.

\bibitem{RosenthalSubspace}
J.~Rosenthal, N.~Silberstein, and A.-L. Trautmann, ``On the geometry of balls in the Grassmannian and list decoding of lifted Gabidulin codes,'' {\em Designs, Codes and Cryptography}, vol.~73, no.~2, pp.~393--416, 2014.
  
\bibitem{MaximumRankArrayCodes}
R.~Roth, ``Maximum-rank array codes and their application to crisscross error correction,'' {\em IEEE Transaction on Information Theory}, vol.~37, pp.~328--336, 2006.
  
\bibitem{SilbersteinRawat-OptimallyLocallyRepairable_2013}
N.~Silberstein, A.~S.~Rawat, O.~O.~Koyluoglu, and S.~Vishwanath, ``{Optimal locally repairable codes via rank-metric codes},'' \emph{IEEE International Symposium on Information Theory (ISIT)}, pp.~1819--1823, 2013.
  
    
\bibitem{SilbersteinRawatVish-ErrorResilDistributedStorage_2012}
N.~Silberstein, A.~S. Rawat, and S.~Vishwanath, ``{Error resilience in distributed storage via rank-metric codes},'' \emph{Allerton Conference on Communication, Control, Computing}, pp. 1150--1157, 2012.


\bibitem{ARankMetricApproach}
D.~Silva, F.~Kschischang, and R.~K\"{o}tter, ``A rank-metric approach to error control in random network coding,'' {\em IEEE Transactions on Information Theory}, vol.~54, no.~9, pp.~3951--3967, 2008.

\bibitem{BoundsOnListDecodingOfRankMetric}
A.~Wachter-Zeh, ``Bounds on list decoding of rank-metric codes,'' {\em IEEE Transactions on Information Theory}, vol.~59, no.~11, pp.~7268--7277, 2013.

\bibitem{xing2003nonlinear}
C.~Xing, ``Nonlinear codes from algebraic curves improving the Tsfasman-Vladut-Zink bound,'' {\em IEEE Transactions on Information Theory}, vol.~49, no.~7, pp.~1653--1657, 2003.
\end{thebibliography}
\end{document}